\documentclass[runningheads]{llncs}
\usepackage[T1]{fontenc}
\usepackage{graphicx}
\usepackage[pdftex]{hyperref}

\usepackage{amssymb} 
\usepackage{graphicx}
\usepackage{enumerate}
\usepackage{xcolor}
\usepackage{comment}

\usepackage{tabularx}
\usepackage{enumitem}
\usepackage{orcidlink}
\usepackage{mathtools}
\usepackage{mathrsfs}
\usepackage{caption}
\usepackage{subcaption}
\usepackage{comment}
\usepackage{complexity}
\usepackage{xargs} 
\usepackage{algorithm}
\usepackage{algorithmic}
\definecolor{MyBlue}{RGB}{50,100,200}
\usepackage{booktabs}
\usepackage{multirow}
\usepackage{pgfplots}
\pgfplotsset{compat=1.17}
\usepackage{thmtools}
\usepackage{thm-restate}
\usepackage{listings}

\hypersetup{%
  colorlinks=true,
  linkcolor=MyBlue,
  citecolor=MyBlue,
  urlcolor=MyBlue
}
\usepackage[nameinlink,capitalise]{cleveref}

\usetikzlibrary{automata, positioning, fit, calc, arrows.meta, patterns}
\usepackage[textsize=tiny]{todonotes}

\usepackage{lineno}

\newcommand{\llbracket}{[\![}
\newcommand{\rrbracket}{]\!]}

\DeclareMathOperator{\Share}{Share}
\DeclareMathOperator{\Reconstruct}{Reconstruct}
\DeclareMathOperator{\Convert}{Convert}

\newcommand{\Mc}{\mathcal{M}}
\newcommand{\Sc}{\mathcal{S}}

\usepackage{todonotes}
\newcommandx{\theju}[2][1=]{\todo[linecolor=blue,backgroundcolor=blue!25,bordercolor=blue,#1]{\tiny T: #2}}
\usepackage{makecell}
\lstdefinelanguage{Python}{
    keywords={def, for, in, range, return, from, import},
    keywordstyle=\color{blue},
    commentstyle=\color{gray},
    stringstyle=\color{orange},
    basicstyle=\ttfamily\small,
    morecomment=[l]{\#},
    morestring=[b]"
}

\begin{document}
\title{Distributed Privacy-Preserving Monitoring}
%
%
\author{Mahyar Karimi\inst{1}\and
K. S. Thejaswini\inst{2}\and
Roderick~Bloem\inst{3}\and\\
Thomas A.~Henzinger\inst{1}}

\authorrunning{M. Karimi et al.}
%
\institute{ISTA, Klosterneuberg, Austria\\
\email{\{mahyar.karimi,tah\}@ista.ac.at}\\
\and 
Université Libre de Bruxelles, Belgium\\
\email{thejaswini.raghavan@ulb.be}\\
\and TU Graz, Austria\\
\email{roderick.bloem@tugraz.at}}
\maketitle              
\begin{abstract}
In traditional runtime verification, a system is
typically observed by a monolithic monitor. Enforcing
privacy in such settings is computationally expensive,
as it necessitates heavy cryptographic primitives. 
Therefore, privacy-preserving monitoring remains  
impractical for real-time applications.

In this work, we address this scalability challenge by
distributing the monitor across multiple parties---at least one of which is honest. This
architecture enables the use of efficient secret-sharing
schemes instead of computationally intensive 
cryptography, dramatically reducing overhead while 
maintaining strong privacy guarantees. While existing 
secret-sharing approaches are typically limited to one-shot executions which do not maintain an internal state, we introduce a protocol tailored for
continuous monitoring that supports repeated evaluations
over an evolving internal state (kept secret from the system and the monitoring entities). We implement our 
approach using the MP-SPDZ framework. Our experiments 
demonstrate that, under these architectural assumptions, 
our protocol is significantly more scalable than 
existing alternatives. 

\end{abstract}
%
%
%
\section{Introduction}
A lot of modern-day software is proprietary. This could include software that is used in healthcare, banks, metro systems, bioinformatic software, malware detection systems, and even verification software. This makes the use of third-party verification a challenge since both the verification software and the data that is being verified in the software need privacy on both ends. Privacy is compromised if protected information about the software, or sensitive data that is processed by it, is revealed. In addition, revealing the specification being verified may undermine the essence of third-party verification. 
Consider for example a bank proving regulatory compliance to auditors without revealing trading algorithms,
or a medical device manufacturer seeking safety certification while protecting control logic.
These highlight a fundamental challenge: verifying compliance with formal specifications while preserving confidentiality of sensitive information.

Unlike static verification which requires the entire model description to determine specification compliance, runtime verification monitors only the specific outputs of a running system. By avoiding the need to analyse the full system state, runtime monitoring offers a more scalable approach to privacy-preserving verification~\cite{LAHPTW22,LJAPW22,LWSTRP24,LJAPW22}. Privacy requirements in most settings of runtime verification are bidirectional. The System contains sensitive data (customer information, trade secrets)
that must remain confidential from the verifier,
while revealing $\varphi$ may undermine independent verification.

\paragraph{Related work.}
Recent approaches to privacy-preserving monitoring largely rely on cryptographic protocols to secure the interaction between a System, which generates a trace of observable outputs, and a Monitor, which holds a specification (e.g., a DFA or STL formula). All of the works consider the privacy-preserving setup where the Monitor does not learn the System's trace, while  the System also does not learn the specification.

Banno et al.~\cite{BMMBWS22} introduce privacy-preserving monitoring using Fully Homomorphic Encryption (FHE), enabling the Monitor to verify Linear Temporal Logic (LTL) specifications privately. Waga et al.~\cite{WMSMBBS24} subsequently extend this approach to Signal Temporal Logic (STL). In these FHE-based frameworks, due to the nature of FHE protocols, the result of the monitoring is revealed solely to the System and cannot be verified by the Monitor. 
Furthermore, Banno et al. and Waga et al. provide security against malicious adversaries, whereas most other works discussed below assume semi-honest (honest-but-curious) adversaries.

Alternatively, Henzinger, Karimi, and Thejaswini~\cite{HKT25} employ Garbled Circuits (GC) and Private Function Evaluation. A key distinction from the FHE-based works is that in this setting, whether the trace is satisfied by the System or not is revealed only to the Monitor. In more recent independent work, Koll, Hang, Rosulek, and Abbas~\cite{koll2025monitoring} also adapt GC protocols to monitor STL specifications for cyber-physical systems. However, Koll et al. require the trace length to be fixed a priori, resulting in linear increase of time required with the pre-determined bound on the trace. In contrast, Henzinger et al. support variable trace lengths that can be determined at runtime by providing a new modification of garbled circuits. 

While these solutions offer strong privacy guarantees, they impose substantial computational costs: the monitoring latency  often requires tens of minutes  to verify even simple specifications due to the complexity of the underlying cryptographic protocols, making them difficult to use for real-time monitoring.

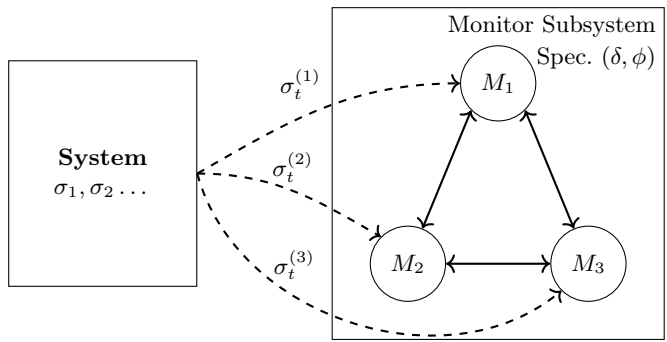
\begin{figure}
    \centering
    \begin{tikzpicture}[
    node distance=0.5cm,
    system/.style={
        draw, rectangle, 
        minimum width=2.5cm, minimum height=3cm, 
        align=center, font=\bfseries
    },
    monitor/.style={
        draw, circle, 
        minimum width=1cm, 
        align=center, font=\small
    },
    container/.style={
        draw, rectangle, 
        inner sep=0.5cm, 
        label={[anchor=north east]north east:Monitor Subsystem},
        label={[anchor=north east,yshift=-0.4cm]north east:Spec. $(\delta,\phi)$}
    }
]

    \node[system] (sys) {System\\$\sigma_1,\sigma_2\dots$};

    \coordinate[right=4cm of sys] (center_point);

    \node[monitor] (m1) at ($(center_point)+(0, 1.2)$) {$M_1$};
    
    \node[monitor] (m2) at ($(center_point)+(-1.2, -1.2)$) {$M_2$};
    
    \node[monitor] (m3) at ($(center_point)+(1.2, -1.2)$) {$M_3$};

    \node[container, fit=(m1) (m2) (m3)] (bigbox) {};

    
    \draw[->, thick,dashed] (sys.east) to[out=30, in=180] node [above,xshift=-0.3cm] {$\sigma_t^{(1)}$}  (m1.west);
    \draw[<->, thick] (m1.south east) to (m3);

    \draw[->, thick, dashed] (sys.east) to[out=0, in=150] node [above] {$\sigma_t^{(2)}$}(m2.north west);
    \draw[<->, thick] (m2) to (m1.south west);

    \draw[->, thick, dashed] (sys.east) to[out=-75, in=220] node [above,xshift=-0.75cm,yshift=0.45cm] {$\sigma_t^{(3)}$}(m3.south west);
    \draw[<->, thick] (m3.west) to (m2.east);

\end{tikzpicture}
    \caption{Monitoring entity that is distributed across several monitors.}
    \label{fig:placeholder}
\end{figure}

\paragraph{Our approach.} We address the 
problem of privacy-preserving monitoring where
a System produces observable outputs (e.g., a sequence of transactions of clients of a bank) verified by a Monitor against a specification. We consider a general setting where the specification may be temporal and depend on an evolving internal state (e.g., the bank balances or transactions of specific client).

We enforce bidirectional privacy: the Monitor learns only whether the specification is satisfied or violated, remaining oblivious to the System's output trace and the internal specification state; conversely, the System learns only the evaluation result, gaining no knowledge of the specification itself.

Unlike previous works that rely on heavy cryptographic primitives, we exploit 
the underlying architecture to achieve privacy and efficiency. We model the Monitor not as a 
monolithic entity, but as a set of distributed parties, where at least one computer is honest. This allows us to leverage \emph{secret 
sharing}, which we explain below, to monitor secret specs on private data. While standard secret-sharing schemes are typically 
restricted to stateless, one-shot functions, we consider a \emph{stateful} protocol for monitoring. This allows for repeated
evaluations of the specification while maintaining a persistent, private internal state.

\paragraph{A tale of two sharings.}
Imagine the System holds a secret byte $\sigma \in \{0,\dots,255\}$ and wants to distribute it among three Monitor computers. 
It simply slices $\sigma$ into three random pieces: it picks $\sigma^{(1)}, \sigma^{(2)}$ at random, computes $\sigma^{(3)}$ such that $\sigma^{(1)} + \sigma^{(2)} + \sigma^{(3)} \equiv \sigma \pmod{256}$, and hands one piece to each computer.
Individually, a share looks like pure noise; even if two computers conspire, they lack the third piece of the puzzle. The secret $\sigma$ exists only when all three combine their shares.
If the System produces a stream of outputs $(x_t, y_t)$, and the Monitor needs to track a weighted running sum like $\sum (5x_t + 8y_t)$, the computers can simply perform the computation locally on their private shares. 
As long as one computer is honest, the individual values remain hidden, and all three monitors can only learn the weighted sum and never the individual values.

Suppose the specification requires multiplying values ($x \cdot y$). 
Here, the additive approach fails, and we turn to the \emph{polynomial} sharing by Shamir~\cite{shamir1979share}.
The System hides secrets $x$ and $y$ by selecting two random degree-2 polynomials, $p$ and $q$ such that $p(0)=x$ and $q(0)=y$.  
It gives each computer a single point on each curve: computer $i$ gets $(i, p(i))$ and $(i, q(i))$.
Since three points uniquely define a parabola, the three computers, but no fewer, can jointly recover the secrets.

The problem arises when they try to multiply. 
Each computer can locally multiply its $y$-coordinates to find a point on the product curve $r(z) = p(z) \cdot q(z)$. 
However, multiplying two quadratic curves produces a \emph{quartic} curve (degree 4). 
A degree-4 curve needs five points to be defined, but we only have three computers! 
The local shares are no longer sufficient to reconstruct the result. 
Resolving this requires additional interactive steps (such as degree 
reduction) which are standard. 



Real-world specifications are rarely just sums or just products.
Verifying a condition like $\sum (x_t^2 + y_t) < 100$ requires polynomial sharing for the product and sum, but Boolean sharing for the comparison. 
While specialised schemes exist for each domain, they 
rely on incompatible data representations: one cannot simply feed an arithmetic 
share directly into a Boolean circuit. 

Modern MPC frameworks like MP-SPDZ~\cite{mpspdz2024} provide mechanisms to 
translate between these representations, but they typically target finite, one-shot 
computations. They do not support the ``infinite,'' non-terminating nature 
of runtime monitoring, where a persistent secret state must be maintained and 
evolved indefinitely across a stream of events.

\paragraph{Our protocol.} 
In our protocol, the System's output trace is secret-shared across the Monitor parties. We depict this pictorially in \cref{fig:placeholder}. 
Post-initialisation, the System and Monitor exchange only a single message per observation step to evaluate the specification, 
thus eliminating expensive cryptographic operations. This yields orders-of-magnitude 
speed-ups compared to FHE or garbled circuits. The key idea is that although the whole monitoring sub-system ``knows'' the secret shared version of the System's observable output, no strict subset of the computers in the Monitor subsystem knows the System's observable output. 

To be able to use different share formats and perform various Boolean and arithmetic operations, our protocol combines multiple secret-sharing techniques.
We utilise share conversion protocols that allow the monitor parties to 
switch between representations: for example, converting arithmetic shares (efficient for summing) into Boolean shares (efficient for comparisons), 
without ever revealing the underlying values. Importantly, this allows 
us to maintain a persistent, private execution state across repeated evaluations, supporting complex specifications that require a mix of 
arithmetic and logical operations. This enables us to monitor much more than invariants but also stateful protocols.

This performance relies on specific assumptions: we require an honest
majority among the distributed Monitor parties. 
We prove that under these constraints, the protocol maintains strong information-theoretic 
privacy: 
no individual party (or dishonest coalition) learns the secret data, and entities learn 
nothing beyond the sequence of specification violations. 

In addition to our protocol, we also provide a
framework for MPC from linear secret-sharing schemes, which can be of independent interest. While Cramer, Damg{\aa}rd, and Maurer~\cite{cramer2000general} also have a similar abstraction, ours
generalises beyond the linear case to a broader class of
sharing systems, which  enables us to reason about protocols that switch
between different sharing representations. 

We evaluate our approach on four diverse monitoring scenarios: access control systems, distributed lock management, blood sugar monitoring, and geofencing. these case studies represent \emph{stateful} temporal specifications 
rather than simple stateless invariants. We adopt the first two benchmarks from 
Henzinger et al.~\cite{HKT25} and the third from Banno et al.~\cite{BMMBWS22}. 
Additionally, we introduce a new geofencing benchmark to demonstrate our 
method's capacity to efficiently handle non-linear operations.
Our implementation achieves several orders of magnitude faster performance than prior cryptographic approaches, with most scenarios completing per-iteration monitoring in well under a second.
This enables practical real-time monitoring while maintaining information-theoretic privacy guarantees.
\section{Problem definition}
\label{sec:definitions}

We formally define the problem and discuss the notation that we use in this paper. 
There are two \emph{entities}, namely the \emph{System}, sometimes referred to as $\Sc$ and the \emph{Monitor} $\Mc$.
The Monitor entity consists of at least two computers, while the System is considered monolithic, and consists of one computer. 
We write  $M_i$ for the $i^\text{th}$ computer in the Monitor entity.
For every pair $(C,D)$ of computers, irrespective of the entity they belong to,
there exists at least one private communication channel between $C$ and~$D$.

\paragraph{System.}
The System is a finite-state machine whose state $\sigma$ is an $n$-bit
binary string.
The System generates the infinite sequence of states
$  \sigma_1, \sigma_2, \dots \in \{0,1\}^n $, with $\sigma_t$ generated at time $t$
This infinite sequence produced by the System is also the \emph{observable output}. 

\paragraph{Monitor.}
The Monitor takes as input a specification in the form of a finite-state
machine with state space $\{0,1\}^m$ and with two functions:
a next-state function $\delta$ and a flag function $\phi$.
All computers in the Monitor entity know $\delta$ and $\phi$.

\begin{itemize}
  \item[--] At time point $t$, the next-state function $\delta$ takes inputs
  $(\mu_t, \sigma_t)$ analogous to those of $\phi$, and outputs the next
  specification state $\mu_{t+1} = \delta(\mu_t, \sigma_t)$.
  \item[--] The flag function $\phi(\mu, \sigma)$ takes as input the current
  \emph{specification state} $\mu \in \{0,1\}^m$ and a System output $\sigma$,
  and returns a Boolean value.
  We write $\phi(\mu, \sigma) = \top$ if and only if the pair $(\mu, \sigma)$ is an
  \emph{undesired input}.
\end{itemize}

\paragraph{Protocols.}
A \emph{protocol} $P$ is a collection of communicating sequential processes.
For each computer $X$ (whether in the System or the Monitor), the protocol
assigns a non-empty process $p(X)$ that runs on $X$.
All computers execute their assigned processes throughout the protocol.

\paragraph{Correctness.}
Protocol $P$ is \emph{correct} if and only if,
at every time point $t$,
the computers of the Monitor entity can collectively obtain the correct values
of $\mu_{t+1}$ and $\phi(\mu_t, \sigma_t)$.
This does not require that any individual computer in the entity must hold
the correct value; correctness only requires that the entity can compute
these values by pooling the information locally available to its computers.
If protocol $P$ is correct, then it terminates at the first time point $t$ at
which $\phi(\mu_t, \sigma_t) = \top$.
\paragraph{Privacy.}
A protocol $P$ is \emph{privacy-preserving} if and only if, at every time step $t$, 
each entity (the System and the Monitor) learns \emph{only} the evaluation result 
$\phi(\mu_t, \sigma_t)$, and nothing regarding the internal state or transitions 
of the other entity beyond what is logically implied by this output. 

In particular, the System must not gain information about the Monitor's state 
$\mu_t$ or transition function $\delta$, and the Monitor must not learn the 
System's output $\sigma_t$ or its transition logic beyond the evaluation result. 
Formally, privacy is defined via \emph{indistinguishability}: for any adversary 
controlling one entity, the view of the protocol execution concerning the other 
entity's secrets is \emph{indistinguishable} from a sequence of 
uniformly random values~\cite{Lin17}.

\paragraph{Repeated executions.} A protocol $P$ satisfies \emph{feasible repeated execution} if  protocol supports repeated execution over the evolving sequence 
    $(\sigma_t)_{t\ge1}$ where each execution is not computationally expensive. 

\paragraph{Assumptions.}

We assume that communication between any two computers is private:
for every pair $(P,Q)$ of computers, $P$ can send a
message to $Q$ such that no computer other than $P$ and $Q$ learns anything about the message.

\paragraph{Security model.}
We work in the \emph{semi-honest} corruption model.
All computers, including corrupted ones, follow the protocol exactly as
specified, but corrupted computers may pool their internal memories in an
attempt to infer additional information.
Security in this model guarantees that such corrupted sets of computers learn
nothing beyond what is already revealed from their own inputs and outputs,
and in particular no unintended information about the internal states of the
other entity is leaked.

\subsubsection{Problem statement.}

Given the System and Monitor entities, their internal state machines, and the
observable outputs $\sigma_t$ as defined in the previous section, our goal is to
design a protocol $P$ that allows these two mutually distrusting entities
to jointly evaluate the specification without revealing any additional
internal information.

At each time point $t$, the Monitor must obtain the correct values of
$\mu_{t+1} = \delta(\mu_t, \sigma_t)$ and the flag $\phi(\mu_t, \sigma_t)$, while the
System and Monitor must not learn anything beyond the value of the flag itself.
Both entities consist of multiple computers operating in the semi-honest
model of adversarial behaviour.

We seek a protocol $P$ that satisfies correctness, privacy and satisfies feasible repeated execution. 




\section{Secret sharing}

The core challenge in privacy-preserving monitoring is to enable verification
without revealing the underlying data.
Secret sharing provides a solution: instead of storing a value in one location,
we distribute it across multiple computers in such a way that no individual
computer learns the value, yet all computers together can reconstruct it and
compute functions on it.

\paragraph{Computing on shares.}
The real power of secret sharing is that we can compute on shared values
\emph{without reconstructing them}.
For example, suppose each computer holds a share of two secret bits $a$ and $b$,
written $\llbracket a \rrbracket$ and $\llbracket b \rrbracket$.
To compute the XOR $c = a \oplus b$ while keeping everything secret,
each computer~$i$ simply computes $c^{(i)} = a^{(i)} \oplus b^{(i)}$ locally.
Since
\begin{align*}
  c^{(1)} \oplus c^{(2)} \oplus c^{(3)} 
  &= (a^{(1)} \oplus b^{(1)}) \oplus (a^{(2)} \oplus b^{(2)}) \oplus (a^{(3)} \oplus b^{(3)}) \\
  &= (a^{(1)} \oplus a^{(2)} \oplus a^{(3)}) \oplus (b^{(1)} \oplus b^{(2)} \oplus b^{(3)}) 
  = a \oplus b = c,
\end{align*}
the computers now hold a valid sharing of $c$ without ever learning $a$, $b$, or $c$.

Computing AND of some bits is more complex: the computers must interact and use additional
randomness to prevent leakage.
This asymmetry---some operations are local, others require interaction---is
fundamental to MPC and drives the need for mixed-protocol computation.

We formulate our protocol in terms of an abstract notion of secret sharing.
This allows us to describe correctness and privacy independently of any
concrete MPC backend or representation~\cite{goldreich1998secure,cramer2000general}.

Our model is organized as follows.
We begin by defining abstract sharing systems (\cref{sec:abstract-sharing}),
which capture the essential properties any secret-sharing scheme must satisfy.
We then introduce share conversion (\cref{sec:share-conversion}) as a
first-class primitive that enables secure interoperability between different
sharing representations.
Finally, we present concrete instantiations (\cref{sec:instantiations})
and establish efficiency criteria for share conversions.

\label{rem:relation-to-prior}
We extend existing frameworks of  Cramer, Damg{\aa}rd, and Maurer~\cite{cramer2000general} 
who provide a
framework for MPC from linear secret-sharing schemes to a broader class of
sharing systems. 
Crucially, we explicitly express share conversion as a primitive
operation with formal correctness and privacy requirements
(Definition~\ref{def:share-conversion}). Treating heterogeneous MPC
backends uniformly under a single interface helps 
us reason about protocols that switch
between different sharing representations.

\subsection{Abstract Sharing Systems}
\label{sec:abstract-sharing}

An abstract sharing system formalises the interface between a secret value
and its distributed representation.
The key idea is to decompose a value $v$ into shares distributed among $k$
computers such that no proper subset of computers learns anything about $v$,
yet all computers together can reconstruct $v$ and compute functions on it.

\begin{definition}[Abstract Sharing System]
\label{def:sharing-system}
An \emph{abstract sharing system} for $k$ computers is a tuple
$\mathcal{S} = (\mathcal{V}, \mathcal{D}, \Share, \Reconstruct, \mathcal{O})$,
where:
\begin{itemize}
  \item $\mathcal{V}$ is a set of secret values.
  \item $\mathcal{D}$ is a domain of shares.
  \item $\Share : \mathcal{V} \to \mathcal{D}^k$ is a (possibly randomised)
        sharing function producing one share per computer.
  \item $\Reconstruct : \mathcal{D}^k \to \mathcal{V}$ is a deterministic
        reconstruction function.
  \item $\mathcal{O}$ is a set of supported operations on shared values,
        where each $f \in \mathcal{O}$ satisfies
        $f \in (\mathcal{V}^r \to \mathcal{V})$ for some $r \in \mathbb{N}$.
\end{itemize}
\end{definition}

Each computer $i \in \{1,\dots,k\}$ holds the $i$th component of $\Share(v)$,
denoted $\Share(v)_i \in \mathcal{D}$.
The sharing function may be probabilistic, so $\Share(v)$ denotes a random
variable over $\mathcal{D}^k$.

The sharing system must satisfy the following properties.

\begin{property}[Correctness]
\label{prop:correctness}
For all $v \in \mathcal{V}$,
$
  \Reconstruct(\Share(v)) = v
$
holds with probability $1$ over the randomness of $\Share$.
\end{property}

\begin{property}[Privacy]
\label{prop:privacy}
For any strict subset $I \subset \{1,\dots,k\}$ of computers,
the joint distribution of shares $\{\Share(v)_i\}_{i \in I}$
is independent of $v$.
Formally, for all $v, v' \in \mathcal{V}$,
$
  \{\Share(v)_i\}_{i \in I} \equiv \{\Share(v')_i\}_{i \in I},
$
where $\equiv$ denotes equality in distribution.
\end{property}

This property ensures perfect privacy: even a coalition of up to $k-1$ computers
learns nothing about the secret beyond what they already knew.

\begin{property}[Closure under Operations]
\label{prop:closure}
For every operation $f \in \mathcal{O}$ of arity $r$,
there exists a protocol $\Pi_f$ that, given sharings
$\llbracket v_1 \rrbracket, \dots, \llbracket v_r \rrbracket$
of values $v_1,\dots,v_r \in \mathcal{V}$,
produces a sharing $\llbracket f(v_1,\dots,v_r) \rrbracket$
without revealing any intermediate values to any computer.
\end{property}

The protocol $\Pi_f$ may involve interaction and randomness, but must
preserve the privacy property: any strict subset of computers learns nothing
about $v_1,\dots,v_r$ or $f(v_1,\dots,v_r)$. 

\paragraph{Notation.}
We write $\llbracket v \rrbracket_{\mathcal{S}}$
to denote a sharing of $v$ under system $\mathcal{S}$,
and omit the subscript when the context is clear.
When we write $\llbracket v \rrbracket$, we mean a tuple of shares
$(s_1,\dots,s_k) \in \mathcal{D}^k$ such that
$\Reconstruct(s_1,\dots,s_k) = v$.

\subsection{Share conversion}
\label{sec:share-conversion}

In practice, different parts of a computation may require different sharing representations.
For example, comparison operations are often more efficient in Boolean form, while arithmetic operations are natural in additive or polynomial sharings.
Share conversion allows the protocol to switch between representations without reconstructing the underlying secret. 
Let
$$
  \mathcal{S}_1 
  = (\mathcal{V}_1,\mathcal{D}_1,\Share_1,\Reconstruct_1,\mathcal{O}_1), \quad
  \mathcal{S}_2 = (\mathcal{V}_2,\mathcal{D}_2,\Share_2,\Reconstruct_2,\mathcal{O}_2)
  $$
be two sharing systems.

\begin{definition}[Share conversion]
\label{def:share-conversion}
A \emph{share conversion} from $\mathcal{S}_1$ to $\mathcal{S}_2$ is a protocol
$
  \Convert_{\mathcal{S}_1 \to \mathcal{S}_2}
  :
  \mathcal{D}_1^k \to \mathcal{D}_2^k
$
such that for every $v \in \mathcal{V}_1 \cap \mathcal{V}_2$:
\begin{enumerate}
  \item \emph{Correctness:}
        reconstructing after conversion yields the same value:\\
        $
          \Reconstruct_2(\Convert(\Share_1(v))) = v
        $
        with probability $1$ over the randomness of $\Share_1$ and $\Convert$.
  \item \emph{Privacy:}
        the protocol reveals no information about $v$
        beyond what is implied by the output sharing.
        Formally, for any strict subset $I \subset \{1,\dots,k\}$ and
        any $v,v' \in \mathcal{V}_1 \cap \mathcal{V}_2$,
        the view of computers $I$ during the execution of $\Convert$
        on input $\Share_1(v)$ is indistinguishable from their view
        on input $\Share_1(v')$.
\end{enumerate}
\end{definition}

We treat share conversion as a first-class primitive.
Its internal realization may involve interaction, randomness,
or auxiliary shared values, but these details are abstracted away.

\begin{remark}[Compositionality]
\label{rem:conversion-composition}
Share conversions compose naturally.
That is, if
$\Convert_{\mathcal{S}_1 \to \mathcal{S}_2}$ and
$\Convert_{\mathcal{S}_2 \to \mathcal{S}_3}$
are share conversions, then their composition
$
  \Convert_{\mathcal{S}_2 \to \mathcal{S}_3} \circ
  \Convert_{\mathcal{S}_1 \to \mathcal{S}_2}
$
is a share conversion from $\mathcal{S}_1$ to $\mathcal{S}_3$.
This follows immediately from the correctness and privacy properties.
\end{remark}

In \cref{app:SharingSystemsConcrete}, we provide concrete
instantiations of the abstract sharing systems using standard MPC schemes:
\emph{additive sharing} over rings $\mathbb{Z}_{2^k}$ for arithmetic computations,
\emph{polynomial (Shamir) sharing} for information-theoretic security in threshold settings, and
\emph{Boolean sharing} for logical predicates.

To support non-linear operations like comparisons ($x < y$) and equality checks,
we use \emph{efficient share conversion} protocols, most notably \emph{bit decomposition},
which translates arithmetic shares into Boolean shares of individual bits.
While linear operations are local and efficient, non-linear operations require interaction:
multiplications utilise precomputed \emph{Beaver triples} to achieve constant-round complexity,
while comparisons rely on conversion protocols followed by optimized Boolean circuits.

\paragraph{Mixed-protocol computation.}
Modern MPC frameworks exploit the efficiency differences between representations:
arithmetic operations (addition, multiplication) are performed on additive or Shamir shares
in $O(1)$ rounds,
Boolean operations (AND, OR, XOR) on Boolean shares in $O(1)$ rounds,
and conversions invoked only when crossing representation boundaries~\cite{demmler2015aby,mohassel2018aby3}.
By carefully choosing representations to minimise conversion overhead,
mixed-protocol computation achieves significantly better concrete efficiency
than using a single representation throughout~\cite{mohassel2018aby3}.

We make the following assumption.
All share conversions required by the protocol admit efficient
realizations as defined in Definition~\ref{def:efficient-conversion}.

\section{Distributed protocol for monitoring}

Having established the abstract sharing framework, we now describe how it
is used to construct a privacy-preserving monitoring protocol.
The key challenge is to enable the Monitor to evaluate the specification
on the System's outputs without learning the outputs themselves,
while ensuring that the System cannot predict or influence what the Monitor learns.

We organize this section as follows.
First, we describe the setup phase and state representation (\cref{sec:setup}).
Then we detail the round-by-round execution (\cref{sec:round-execution})
and formalize the properties that ensure privacy and correctness
(\cref{sec:invariants}).


\subsection{Setup and state representation}
\label{sec:setup}

Before execution begins, the System and Monitor establish a common framework
for secret sharing and agree on how the specification will be represented.

\paragraph{Shared parameters.}
The System and Monitor agree on a fixed sharing system
$\mathcal{S} = (\mathcal{V}, \mathcal{D}, \Share, \Reconstruct, \mathcal{O})$.
This agreement includes concrete definitions
for each element in the tuple; e.g.,
integers modulo $2^{64}$ or bits as the value domain $\mathcal{V}$, or
arithmetic or logical operations or comparisons
as the set of supported operations $\mathcal{O}$.

All sharings exchanged during the protocol are instances of $\mathcal{S}$.
This ensures that values shared by the System can be directly consumed by
the Monitor in its local computations without additional conversions between
incompatible representations.
Note that the Monitor computers can use other agreed-upon sharing systems
among themselves as well, to perform share conversions.
\paragraph{Persistent shared state.}
The Monitor entity maintains the specification state $\mu_t$ in shared form
throughout execution.
At every time point $t$, the following invariant holds:
\begin{quote}
  The Monitor computers jointly hold a sharing
  $\llbracket \mu_t \rrbracket$ under the secret sharing system $\mathcal{S}$,
  and no individual Monitor computer learns $\mu_t$.
\end{quote}
This is the central privacy mechanism: even though the Monitor collectively
knows the current specification state, no single Monitor computer can extract
it without cooperation from the others.

In contrast, the System maintains its internal state $\sigma_t$ locally and
does not share it with the Monitor.
The System's state may include sensitive information such as cryptographic keys,
user data, or proprietary algorithms.
The protocol ensures that this state remains entirely private.

\subsection{Round-by-round execution}
\label{sec:round-execution}

Execution proceeds in discrete rounds indexed by $t = 1,2,\dots$.
Each round transforms the shared Monitor state from $\llbracket \mu_t \rrbracket$
to $\llbracket \mu_{t+1} \rrbracket$ and produces a single violation flag bit.

\paragraph{Round $t$ consists of three phases.}

\begin{enumerate}
        

  \item \textbf{Sharing of observable output.}
        The System computes a sharing
        $
          \llbracket \sigma_t \rrbracket = \Share(\sigma_t)
        $
        under the agreed sharing system $\mathcal{S}$.
        The System sends the $i^\text{th}$ share of $\llbracket \sigma_t \rrbracket$
        to the $i^\text{th}$ Monitor computer using a private channel.
        
        After this step, the Monitor entity jointly holds a sharing of $\sigma_t$,
        while no individual Monitor computer learns any information about
        $\sigma_t$ beyond its own share.
        By Property~\ref{prop:privacy}, any strict subset of Monitor computers
        learns nothing about $\sigma_t$.

  \item \textbf{Specification evaluation.}
        Using the shared inputs $\llbracket \mu_t \rrbracket$ (from the
        previous round) and $\llbracket \sigma_t \rrbracket$ (from the current round),
        the Monitor computers jointly evaluate the specification.
        
        More precisely, they execute a secure multiparty computation that produces:
        \begin{itemize}
          \item a shared violation flag
                $\llbracket \phi(\mu_t, \sigma_t) \rrbracket \in \llbracket \{0,1\} \rrbracket$, and
          \item a shared encoding of the next specification state
                $\llbracket \mu_{t+1} \rrbracket
                  = \llbracket \delta(\mu_t, \sigma_t) \rrbracket$.
        \end{itemize}
        
        This computation is expressed entirely in terms of operations and
        conversions supported by $\mathcal{S}$.
        For example, evaluating a comparison $\sigma_t > c$ may require converting
        $\llbracket \sigma_t \rrbracket$ from arithmetic to Boolean form,
        performing a bitwise comparison, and converting the result back.
        All intermediate values remain secret-shared throughout---no
        intermediate result is ever reconstructed.

  \item \textbf{Flag reconstruction and control flow.}
        The Monitor computers jointly reconstruct the shared bit
        $\llbracket \phi(\mu_t, \sigma_t) \rrbracket$ using the $\Reconstruct$ function.
        This is the \emph{only} value reconstructed in each round.
        
        If the reconstructed value is $\top$ (a violation),
        the protocol terminates and reports a violation to an external observer.
        Otherwise, the reconstructed flag is discarded, and the Monitor
        retains $\llbracket \mu_{t+1} \rrbracket$ as the persistent shared
        state for the next round.
        The specification state $\mu_{t+1}$ is never revealed.
\end{enumerate}


\subsection{Invariants and information flow}
\label{sec:invariants}

We now formalize the key properties that hold throughout execution.

\begin{property}[State evolution invariant]
\label{prop:state-invariant}
At the beginning of each round $t$, the Monitor entity jointly holds
$\llbracket \mu_t \rrbracket$, and no other information about the execution
history is stored in the clear.
The only values ever reconstructed during the protocol are the violation flags
$\phi(\mu_1, \sigma_1), \phi(\mu_2, \sigma_2), \ldots$ at each round.
\end{property}

This invariant ensures that the protocol's information leakage is minimal and
explicit: one bit per round, and nothing else.

\subsection{Computation model}
\label{sec:computation-model}

To make the specification evaluation phase concrete, we describe the
computational substrate on which the Monitor evaluates the specification.
The key challenge is to explicitly track where share conversions occur and
ensure that all operations are type-safe with respect to the sharing representation.

We adopt a typed instruction set architecture with register-based memory,
similar to the approach used in practical MPC systems such as
MP-SPDZ~\cite{keller2020mp}.
While this design is inspired by existing MPC frameworks, we formalise it in \cref{app:Register}
and provide a precise operational semantics for our protocol and to make the
cost of share conversions explicit in the analysis.

To model the execution of privacy-preserving monitoring, we formalise the Monitor as an abstract machine operating on a finite set of typed registers, where each register $R_i$  holds a value secret-shared across distributed parties. We define a specific instruction set that manipulates these registers without reconstructing secrets, enforcing a strict type discipline: Arithmetic operations ($\mathtt{ADD}$, $\mathtt{MUL}$) operate on arithmetic shares, Boolean operations ($\mathtt{AND}$, $\mathtt{XOR}$) on Boolean shares, and Comparison operations ($\mathtt{LT}$, $\mathtt{EQ}$) serve as the bridge between the two. Comparison instructions internally perform the expensive share conversion protocols necessary to translate arithmetic inputs into Boolean outputs. This abstraction allows us to compile high-level specifications into well-typed instruction sequences. It further helps us directly compute the computational cost of share conversions.

\begin{property}[Information flow from System to Monitor]
\label{prop:info-flow-system}
All information flow from the System to the Monitor occurs exclusively via
fresh sharings of observable outputs $\sigma_t$.
More explicitly, the System never receives information derived from $\llbracket \mu_t \rrbracket$
or from intermediate Monitor computations.
\end{property}

This property ensures that the System cannot learn what the Monitor is checking
or adapt its behaviour based on the specification state.

\begin{property}[Information flow from Monitor to outside]
\label{prop:info-flow-monitor}
All information flow from the Monitor to the outside world occurs exclusively
via reconstruction of violation flags $\phi(\mu_t, \sigma_t)$.
In particular:
  individual Monitor computers never receive information about $\sigma_t$
        beyond what is implied by their shares of $\sigma_t$ and the reconstructed flag, and 
  no information about $\mu_t$ or intermediate values is ever revealed
        to any party.
\end{property}

Together, Properties~\ref{prop:info-flow-system} and~\ref{prop:info-flow-monitor}
formalize the separation between System privacy and Monitor privacy.

\begin{restatable}[Compositionality across rounds]{theorem}{compositionality}
\label{thm:compositionality}
Privacy is preserved under unbounded repetition of the protocol.
\end{restatable}
Each round of the protocol is independent except for the shared state
$\llbracket \mu_t \rrbracket$ carried forward by the Monitor.
Because no additional information is revealed between rounds beyond the
violation flags, we show that privacy is preserved under unbounded repetition of the protocol.
This compositionality property is essential for long-running Systems where
the number of monitoring rounds is not fixed in advance.



\subsection{Adversary model and privacy guarantees}

Recall from Section~\ref{sec:definitions} that we work in the semi-honest
corruption model:
corrupted computers follow the protocol specification exactly but may
pool their internal state to infer additional information.
We consider static corruption, where the set of corrupted computers is fixed
before execution begins.

\begin{definition}[Adversary structure]
\label{def:adversary}
An \emph{adversary structure} is a set $(I_M)$ where
$I_M \subseteq \{M_1,\dots,M_{k_M}\}$ is the set of corrupted
computers in the Monitor entity.
We require $|I_M| < k_M$ (i.e., each entity has
at least one honest computer).
\end{definition}

The adversary controlling $I_M$ observes all messages sent to or from
corrupted computers in $I_M$, all local computations performed by corrupted computers,
and all random coins used by corrupted computers.
The adversary may not observe communication between honest computers or
their internal states.


\begin{restatable}[System privacy]{theorem}{systemprivacy}
\label{thm:system-privacy}
For any adversary structure $(I_M)$ with $|I_M| < k_M$,
the view of corrupted Monitor computers reveals nothing about the System's
internal state $\sigma_t$ or observable outputs $x_t$
beyond what is implied by the sequence of reconstructed violation flags
$\phi(\mu_1, x_1), \dots, \phi(\mu_T, x_T)$.
\end{restatable}

\begin{restatable}[Monitor privacy]{theorem}{monitorprivacy}
\label{thm:monitor-privacy}
For any adversary structure $(I_M)$ with $|I_M| < k_M$,
the view of corrupted System computers reveals nothing about the Monitor's
specification state $\mu_t$, the next-state function $\delta$,
or the flag function $\phi$,
beyond what is implied by the sequence of reconstructed violation flags
$\phi(\mu_1, x_1), \dots, \phi(\mu_T, x_T)$.
\end{restatable}

\begin{corollary}[Joint privacy]
\label{cor:joint-privacy}
Under the semi-honest model with adversary structure $(I_M)$ satisfying
$|I_M| < k_M$,
the protocol reveals to the adversary only the sequence of violation flags.
\end{corollary}

We end with a remark that the privacy guarantees of our protocols depend on the underlying sharing system. If the sharing system provides perfect (information-theoretic)
 privacy---as is the case for additive sharing, polynomial (Shamir) sharing,
 and Boolean XOR sharing---then our protocol achieves perfect privacy against
 computationally unbounded adversaries. However, if the sharing system or share conversion protocols rely on computational
 assumptions (e.g., garbled circuits for Boolean operations, or
 lattice-based conversions), then privacy holds only against
 polynomially-bounded adversaries.

\section{Experiments}\label{sec:experiments}

We implemented and evaluated our privacy-preserving monitoring protocol
using the MP-SPDZ framework~\cite{keller2020mp},
specifically using Shamir secret sharing with external I/O capabilities.
The implementation artifacts and experiment scripts are available in our
fork of MP-SPDZ, on the \texttt{experiments} branch.\footnote{\url{https://github.com/mahykari/MP-SPDZ/tree/experiments/}}
Our experimental setup demonstrates the practical feasibility of the
approach for real-time monitoring scenarios.

\subsection{Implementation details}

\paragraph{MPC framework.}
We use MP-SPDZ's Shamir secret sharing protocol,
which provides information-theoretic security for privacy guarantees.
Unlike computation-heavy cryptographic approaches,
Shamir secret sharing operates in finite fields and does not rely on
computational hardness assumptions.
A template moonitoring program in this framework is presented
in Appendix~\ref{app:monitor-template}.

\paragraph{Implementation architecture.}
Our implementation leverages MP-SPDZ's client interface to separate
the roles of data providers (System entities) and computing parties (Monitor entities).
Data providers contribute secret-shared inputs without participating in computation,
while Monitor computers execute the monitoring protocol without learning individual System secrets.
The implementation involved working with MP-SPDZ's primitives for secret-shared operations,
including comparisons and conditional evaluations that arise in temporal logic specifications.
Understanding how to efficiently compose these primitives—particularly for
operations like threshold comparisons and Boolean combinations of predicates—required
familiarity with the framework's approach to field arithmetic and share conversions.

\paragraph{Security configuration.}
All experiments use a 128-bit field with 3 parties and
corruption threshold 1;
i.e., no two Monitor computers pool their information.

\paragraph{Network security.}
Party-to-party communication in MP-SPDZ uses TLS encryption
by default via the \texttt{CryptoPlayer} component,
providing confidentiality, authentication, and integrity.
The System uses X.509 certificate-based authentication with TLS 1.2,
establishing two separate TLS connections per party pair
(one for sending, one for receiving) to handle bidirectional
communication.

\subsection{Experimental results}

We evaluate the performance and scalability of our approach using
the two case studies from Henzinger et al.~\cite{HKT25}:
an access control system (ACS) with bidirectional doors,
and a distributed lock management system.
Moreover, we evaluate on a specification
proposed by Banno et al. \cite{BMMBWS22},
as well as a novel scenario that showcases
a more computationally involved specification.
All experiments use a 128-bit field
(except for the specification from Banno et al.),
with 3 computers for the Monitor entity,
and run for multiple iterations to measure steady-state performance.

\subsubsection{Access Control System (ACS).}

The ACS case study monitors an office building with two types of employees
(types A and B) entering and exiting through multiple external doors.
Each door tracks four values: entries and exits for type A and type B employees.
The specification requires that the count of type A employees currently in the
building never falls below the count of type B employees.
We evaluate scalability from 10 to 1000 doors,
with Henzinger et al.~\cite{HKT25} having evaluated only
up to 30 doors in their original work.

\subsubsection{Distributed lock management.}

The lock management case study models a parallel program with multiple locks,
where each lock can be in one of two states (LOCK or UNLOCK).
At each round, the Monitor receives requests to lock, unlock, or skip each lock.
The specification ensures that lock() or unlock() is never called twice in a row
for any lock.
This scenario is inherently more Boolean-heavy than ACS,
requiring extensive Boolean operations to check state transitions.
We evaluate settings with 100 to 1000 locks.

\subsubsection{Presidential car geofence scenario.}

The Presidential car scenario demonstrates a more computationally intensive specification.
The Monitor tracks the position of a car in an $n$-dimensional integer space ($\mathbb{Z}^n$),
receiving displacement vectors at each round.
The specification requires that the car remains within an expanding sphere:
the sphere radius starts at a base level, grows linearly with time,
and stops at a maximum radius.
This scenario exercises the arithmetic capabilities of our MPC framework,
requiring non-linear arithmetic and comparisons in high-dimensional spaces.

\subsubsection{Blood sugar monitoring.}

The Blood sugar scenario, proposed by Banno et al.~\cite{BMMBWS22},
monitors a blood sugar sensor that transmits readings regularly.
We evaluate the specification $\Box_{[600,700]} (x \leq 200)$,
which requires that within any time window of length between 600 and 700 time units,
the blood sugar level never exceeds 200.
Banno et al. measured multiple specifications; we selected this one as
representative since others show similar characteristics. 

\subsection{Performance analysis}

\paragraph{Circuit complexity.}
Table~\ref{tab:combined-metrics} shows the offline preprocessing requirements
for all four scenarios.
For ACS, the number of integer triples scales linearly with the number of doors,
growing from 4,500 triples for 10 doors to 400,500 for 1000 doors.
Notably, the bit triple count (used for Boolean operations) and daBit count
remain constant,
indicating that the circuit depth for Boolean operations
is independent of the number of doors.
This is expected for our parallel monitoring architecture where each
door's computation is independent.

The lock system requires significantly more preprocessing than ACS due to its Boolean-heavy nature.
The triple count scales linearly from 60,100 (100 locks) to 600,100 (1000 locks).
Unlike ACS, the bit triple count and daBit count also scale with system size,
reflecting the extensive Boolean operations required for lock state checking.

The Presidential car scenario scales with dimensionality.
For the 4-dimensional case, we require 3,800 triples,
growing to 819,800 triples for the 1,024-dimensional case.
Like ACS, the bit triple and daBit counts remain constant,
as the core computation is arithmetic rather than Boolean.

The Blood sugar scenario requires minimal preprocessing (only 2 triples and 8 daBits),
reflecting the straightforward nature of the specification,
which primarily involves threshold comparisons over a sliding window.
\begin{table}[h]
\centering
\caption{Circuit complexity and Performance metrics}
\label{tab:combined-metrics}
\setlength{\tabcolsep}{4pt} 
\begin{tabular*}{\textwidth}{@{\extracolsep{\fill}}lrrr rrr@{}}
\toprule
 & \multicolumn{3}{c}{\textbf{Circuit Complexity}} & \multicolumn{3}{c}{\textbf{Performance}} \\
 \cmidrule(lr){2-4} \cmidrule(l){5-7}
\textbf{Scenario} & 
\textbf{Triples} & 
\makecell[b]{\textbf{Bit-}\\\textbf{Triples}} & 
\textbf{daBits} & 
\makecell[b]{\textbf{Comp.}\\\textbf{(s)}} & 
\makecell[b]{\textbf{Total}\\\textbf{Time (s)}} & 
\makecell[b]{\textbf{Data}\\\textbf{(MB)}} \\
\midrule
\multicolumn{7}{@{}l}{\textit{ACS (Doors)}} \\
\phantom{1,0}10 doors    & 4,500   & 247 & 1 & 0.030 & 0.070 & 0.013 \\
\phantom{1,0}30 doors    & 12,500  & 247 & 1 & 0.037 & 0.070 & 0.030 \\
\phantom{1,}100 doors    & 40,500  & 247 & 1 & 0.031 & 0.084 & 0.076 \\
\phantom{1,}300 doors    & 120,500 & 247 & 1 & 0.030 & 0.106 & 0.180 \\
1,000 doors              & 400,500 & 247 & 1 & 0.024 & 0.177 & 0.583 \\
\midrule
\multicolumn{7}{@{}l}{\textit{Locks}} \\
\phantom{1,}100 locks    & 60,100  & 12,700  & 100   & 0.151 & 0.161 & 0.357 \\
\phantom{1,}300 locks    & 180,100 & 38,100  & 300   & 0.412 & 0.432 & 1.070 \\
\phantom{1,}500 locks    & 300,100 & 63,500  & 500   & 0.590 & 0.613 & 1.783 \\
1,000 locks              & 600,100 & 127,000 & 1,000 & 1.230 & 1.270 & 3.564 \\
\midrule
\multicolumn{7}{@{}l}{\textit{Presidential Car ($n$--D: $n$--dimensional)}} \\
\phantom{1,02}4--D       & 3,800   & 247 & 1 & 0.061 & 0.063 & 0.956 \\
\phantom{1,0}16--D       & 13,400  & 247 & 1 & 0.060 & 0.064 & 1.269 \\
\phantom{1,0}64--D       & 51,800  & 247 & 1 & 0.063 & 0.069 & 2.535 \\
\phantom{1,}256--D       & 205,400 & 247 & 1 & 0.054 & 0.071 & 7.295 \\
1,024--D                 & 819,800 & 247 & 1 & 0.091 & 0.176 & 25.592 \\
\midrule
\multicolumn{7}{@{}l}{\textit{Blood Sugar}} \\
$\Box_{[600,700]} \left( x \leq 200 \right)$ & 2 & 0 & 8 & 0.000 & 0.078 & 0.002 \\
\bottomrule
\end{tabular*}
\end{table}
\paragraph{Performance metrics.}
On the right, Table~\ref{tab:combined-metrics} presents per-iteration performance metrics for all four scenarios.

For ACS, the computation time remains remarkably stable (0.024--0.037s) regardless
of System size, while communication overhead scales linearly.
This demonstrates the efficiency of our MPC operations,
with the bottleneck being I/O and network coordination rather than computation.

The lock system is slower than ACS due to the Boolean-heavy nature of lock state checking,
but still maintains sublinear scaling:
10× more locks results in only 7.9× increase in latency.
Computation time scales with system size but remains practical,
ranging from 0.151s to 1.230s.

For the Presidential Car scenario, computation time remains stable
across different dimensionalities (0.054--0.091s),
with the exception of the highest dimensional case (1,024-D) where
communication overhead becomes more significant.
Despite the high-dimensional computations,
the per-iteration time remains well under 200ms even for 1,024 dimensions. 

The Blood Sugar scenario executes extremely quickly with negligible computation time
and 0.078s total time per iteration. 

\begin{figure}[h]
\centering
\begin{minipage}{0.32\textwidth}
\centering
\begin{tikzpicture}
\begin{axis}[
    xbar,
    y axis line style={opacity=0},
    axis x line*=bottom,
    xtick pos=left,
    ytick pos=left,
    ytick=data,
    yticklabels={10, 30, 100, 300, 1000},
    xlabel={Time (s)},
    ylabel={ACS (Doors)},
    ylabel style={font=\tiny},
    width=\textwidth,
    height=0.18\textheight,
    xmin=0,
    xmax=0.2,
    bar width=6pt,
    enlarge y limits={abs=0.5},
    nodes near coords,
    nodes near coords align={horizontal},
    nodes near coords style={font=\tiny,/pgf/number format/fixed,/pgf/number format/precision=3},
    every node near coord/.append style={anchor=west, xshift=1pt},
    scaled x ticks=false,
    major tick length=2pt,
    xticklabel style={/pgf/number format/fixed,font=\tiny},
    yticklabel style={font=\tiny},
    xlabel style={font=\tiny},
]
\addplot[fill=blue!60] coordinates {
    (0.070, 0)
    (0.070, 1)
    (0.084, 2)
    (0.106, 3)
    (0.177, 4)
};
\end{axis}
\end{tikzpicture}
\end{minipage}%
\hfill%
\begin{minipage}{0.32\textwidth}
\centering
\begin{tikzpicture}
\begin{axis}[
    xbar,
    y axis line style={opacity=0},
    axis x line*=bottom,
    xtick pos=left,
    ytick pos=left,
    ytick=data,
    yticklabels={10, 30, 100, 300, 1000},
    xlabel={Triples ($\times 10^5$)},
    width=\textwidth,
    height=0.18\textheight,
    xmin=0,
    bar width=6pt,
    enlarge y limits={abs=0.5},
    nodes near coords,
    nodes near coords align={horizontal},
    nodes near coords style={font=\tiny,/pgf/number format/fixed,/pgf/number format/precision=2},
    every node near coord/.append style={anchor=west, xshift=1pt},
    scaled x ticks=false,
    major tick length=2pt,
    xticklabel style={/pgf/number format/fixed,font=\tiny},
    yticklabel style={font=\tiny},
    xlabel style={font=\tiny},
]
\addplot[fill=gray!70] coordinates {
    (0.045, 0)
    (0.125, 1)
    (0.405, 2)
    (1.205, 3)
    (4.005, 4)
};
\end{axis}
\end{tikzpicture}
\end{minipage}%
\hfill%
\begin{minipage}{0.32\textwidth}
\centering
\begin{tikzpicture}
\begin{axis}[
    xbar,
    y axis line style={opacity=0},
    axis x line*=bottom,
    xtick pos=left,
    ytick pos=left,
    ytick=data,
    yticklabels={10, 30, 100, 300, 1000},
    xlabel={Data (MB)},
    width=\textwidth,
    height=0.18\textheight,
    xmin=0,
    bar width=6pt,
    enlarge y limits={abs=0.5},
    nodes near coords,
    nodes near coords align={horizontal},
    nodes near coords style={font=\tiny,/pgf/number format/fixed,/pgf/number format/precision=3},
    every node near coord/.append style={anchor=west, xshift=1pt},
    scaled x ticks=false,
    major tick length=2pt,
    xticklabel style={/pgf/number format/fixed,font=\tiny},
    yticklabel style={font=\tiny},
    xlabel style={font=\tiny},
]
\addplot[fill=green!60] coordinates {
    (0.013, 0)
    (0.030, 1)
    (0.076, 2)
    (0.180, 3)
    (0.583, 4)
};
\end{axis}
\end{tikzpicture}
\end{minipage}

\vspace{0.3cm}

\begin{minipage}{0.32\textwidth}
\centering
\begin{tikzpicture}
\begin{axis}[
    xbar,
    y axis line style={opacity=0},
    axis x line*=bottom,
    xtick pos=left,
    ytick pos=left,
    ytick=data,
    yticklabels={100, 300, 500, 1000},
    xlabel={Time (s)},
    ylabel={Locks},
    ylabel style={font=\tiny},
    width=\textwidth,
    height=0.18\textheight,
    xmin=0,
    bar width=6pt,
    enlarge y limits={abs=0.5},
    nodes near coords,
    nodes near coords align={horizontal},
    nodes near coords style={font=\tiny,/pgf/number format/fixed,/pgf/number format/precision=3},
    every node near coord/.append style={anchor=west, xshift=1pt},
    scaled x ticks=false,
    major tick length=2pt,
    xticklabel style={/pgf/number format/fixed,font=\tiny},
    yticklabel style={font=\tiny},
    xlabel style={font=\tiny},
]
\addplot[fill=blue!60] coordinates {
    (0.151, 0)
    (0.432, 1)
    (0.613, 2)
    (1.270, 3)
};
\end{axis}
\end{tikzpicture}
\end{minipage}%
\hfill%
\begin{minipage}{0.32\textwidth}
\centering
\begin{tikzpicture}
\begin{axis}[
    xbar,
    y axis line style={opacity=0},
    axis x line*=bottom,
    xtick pos=left,
    ytick pos=left,
    ytick=data,
    yticklabels={100, 300, 500, 1000},
    xlabel={Triples ($\times 10^5$)},
    width=\textwidth,
    height=0.18\textheight,
    xmin=0,
    bar width=6pt,
    enlarge y limits={abs=0.5},
    nodes near coords,
    nodes near coords align={horizontal},
    nodes near coords style={font=\tiny,/pgf/number format/fixed,/pgf/number format/precision=2},
    every node near coord/.append style={anchor=west, xshift=1pt},
    scaled x ticks=false,
    major tick length=2pt,
    xticklabel style={/pgf/number format/fixed,font=\tiny},
    yticklabel style={font=\tiny},
    xlabel style={font=\tiny},
]
\addplot[fill=gray!70] coordinates {
    (0.601, 0)
    (1.801, 1)
    (3.001, 2)
    (6.001, 3)
};
\end{axis}
\end{tikzpicture}
\end{minipage}%
\hfill%
\begin{minipage}{0.32\textwidth}
\centering
\begin{tikzpicture}
\begin{axis}[
    xbar,
    y axis line style={opacity=0},
    axis x line*=bottom,
    xtick pos=left,
    ytick pos=left,
    ytick=data,
    yticklabels={100, 300, 500, 1000},
    xlabel={Data (MB)},
    width=\textwidth,
    height=0.18\textheight,
    xmin=0,
    bar width=6pt,
    enlarge y limits={abs=0.5},
    nodes near coords,
    nodes near coords align={horizontal},
    nodes near coords style={font=\tiny,/pgf/number format/fixed,/pgf/number format/precision=3},
    every node near coord/.append style={anchor=west, xshift=1pt},
    scaled x ticks=false,
    major tick length=2pt,
    xticklabel style={/pgf/number format/fixed,font=\tiny},
    yticklabel style={font=\tiny},
    xlabel style={font=\tiny},
]
\addplot[fill=green!60] coordinates {
    (0.357, 0)
    (1.070, 1)
    (1.783, 2)
    (3.564, 3)
};
\end{axis}
\end{tikzpicture}
\end{minipage}

\vspace{0.3cm}

\begin{minipage}{0.32\textwidth}
\centering
\begin{tikzpicture}
\begin{axis}[
    xbar,
    y axis line style={opacity=0},
    axis x line*=bottom,
    xtick pos=left,
    ytick pos=left,
    ytick=data,
    yticklabels={4, 16, 64, 256, 1024},
    xlabel={Time (s)},
    ylabel={Presidential Car},
    ylabel style={font=\tiny},
    width=\textwidth,
    height=0.18\textheight,
    xmin=0,
    xmax=0.2,
    bar width=6pt,
    enlarge y limits={abs=0.5},
    nodes near coords,
    nodes near coords align={horizontal},
    nodes near coords style={font=\tiny,/pgf/number format/fixed,/pgf/number format/precision=3},
    every node near coord/.append style={anchor=west, xshift=1pt},
    scaled x ticks=false,
    major tick length=2pt,
    xticklabel style={/pgf/number format/fixed,font=\tiny},
    yticklabel style={font=\tiny},
    xlabel style={font=\tiny},
]
\addplot[fill=blue!60] coordinates {
    (0.063, 0)
    (0.064, 1)
    (0.069, 2)
    (0.071, 3)
    (0.176, 4)
};
\end{axis}
\end{tikzpicture}
\end{minipage}%
\hfill%
\begin{minipage}{0.32\textwidth}
\centering
\begin{tikzpicture}
\begin{axis}[
    xbar,
    y axis line style={opacity=0},
    axis x line*=bottom,
    xtick pos=left,
    ytick pos=left,
    ytick=data,
    yticklabels={4, 16, 64, 256, 1024},
    xlabel={Triples ($\times 10^5$)},
    width=\textwidth,
    height=0.18\textheight,
    xmin=0,
    bar width=6pt,
    enlarge y limits={abs=0.5},
    nodes near coords,
    nodes near coords align={horizontal},
    nodes near coords style={font=\tiny,/pgf/number format/fixed,/pgf/number format/precision=2},
    every node near coord/.append style={anchor=west, xshift=1pt},
    scaled x ticks=false,
    major tick length=2pt,
    xticklabel style={/pgf/number format/fixed,font=\tiny},
    yticklabel style={font=\tiny},
    xlabel style={font=\tiny},
]
\addplot[fill=gray!70] coordinates {
    (0.038, 0)
    (0.134, 1)
    (0.518, 2)
    (2.054, 3)
    (8.198, 4)
};
\end{axis}
\end{tikzpicture}
\end{minipage}%
\hfill%
\begin{minipage}{0.32\textwidth}
\centering
\begin{tikzpicture}
\begin{axis}[
    xbar,
    y axis line style={opacity=0},
    axis x line*=bottom,
    xtick pos=left,
    ytick pos=left,
    ytick=data,
    yticklabels={4, 16, 64, 256, 1024},
    xlabel={Data (MB)},
    width=\textwidth,
    height=0.18\textheight,
    xmin=0,
    bar width=6pt,
    enlarge y limits={abs=0.5},
    nodes near coords,
    nodes near coords align={horizontal},
    nodes near coords style={font=\tiny,/pgf/number format/fixed,/pgf/number format/precision=3},
    every node near coord/.append style={anchor=west, xshift=1pt},
    scaled x ticks=false,
    major tick length=2pt,
    xticklabel style={/pgf/number format/fixed,font=\tiny},
    yticklabel style={font=\tiny},
    xlabel style={font=\tiny},
]
\addplot[fill=green!60] coordinates {
    (0.956, 0)
    (1.269, 1)
    (2.535, 2)
    (7.295, 3)
    (25.592, 4)
};
\end{axis}
\end{tikzpicture}
\end{minipage}
\caption{Performance comparison across all case studies: (left) per-iteration timing, (center) circuit complexity (number of triples), (right) communication overhead. The three rows show ACS, Locks, and Presidential Car scenarios respectively.}
\label{fig:all-performance}
\end{figure}
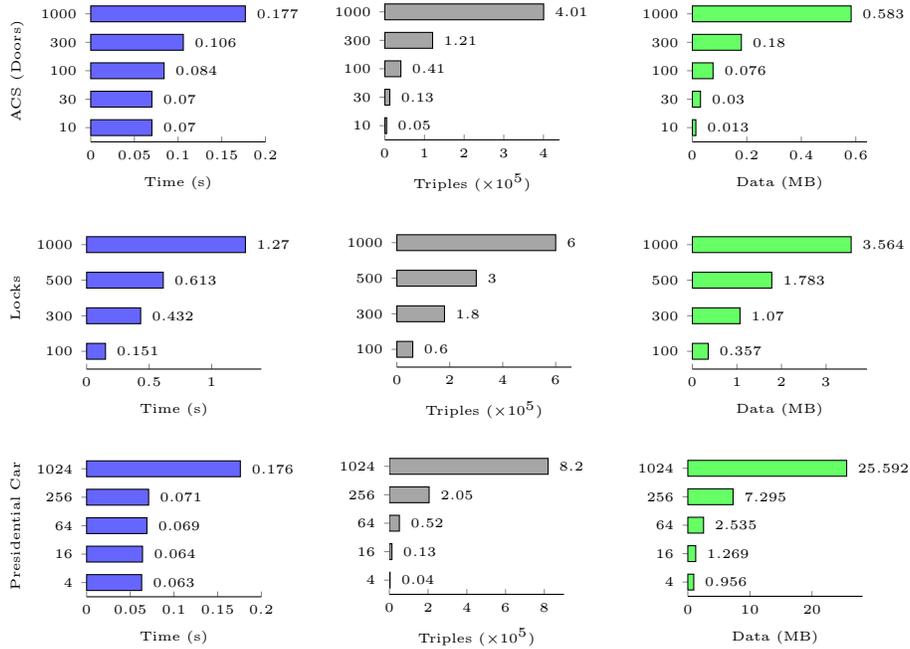

\subsubsection{Comparison with prior work}
\paragraph{Performance improvement.}
The prior work by Henzinger et al.~\cite{HKT25} using
garbled circuits with an industrial standard security parameter of 1536 reports execution times
of \emph{at least} 18 seconds per iteration for comparable specifications.
As specification complexity increases,
their overhead grows to several minutes per iteration.
The approach of Banno et al.~\cite{BMMBWS22} is even more constrained; 
as noted by Henzinger et al., it renders benchmarks like ACS and Locks 
infeasible due to the prohibitive size of the DFA
required to represent them. Similarly, it fails to handle the geofencing 
(Presidential car) scenario again due to similar reasons. On the blood glucose benchmarks---which 
Banno et al. can process---our implementation matches the faster of the two protocols they consider, completing each iteration in a fraction of a second.

In contrast, our approach achieves:
\begin{itemize}
    \item \textbf{ACS:} 0.07--0.18s per iteration (100--250× faster)
    \item \textbf{Locks:} 0.16--1.3s per iteration (14--112× faster)
    \item \textbf{Presidential Car:} 0.06--0.18s per iteration
    \item \textbf{Blood Sugar:} 0.08s per iteration
\end{itemize}

This represents a \emph{2--3 orders of magnitude improvement} in
execution time while maintaining information-theoretic security guarantees
that do not rely on computational hardness assumptions.

\paragraph{Scalability Characteristics.}
Figure~\ref{fig:all-performance} illustrates
the performance characteristics of the case studies.
For ACS and Presidential car, timing remains nearly constant while preprocessing requirements
and communication scale linearly.
For the lock system, all metrics scale with system size
due to the Boolean-heavy nature of lock state checking.
Despite this, our Shamir secret sharing approach maintains practical
performance that is 2--3 orders of magnitude faster than cryptographic
approaches, making real-time monitoring of large-scale systems practical.


\section{Discussion}
We have presented a privacy-preserving monitoring protocol that significantly outperforms existing alternatives in terms of scalability. This efficiency is achieved by leveraging an architectural distribution of trust: we replace heavy cryptographic primitives with lightweight secret sharing, which requires an honest majority assumption (i.e., at least one monitor party must remain uncompromised). This represents a fundamental trade-off, where we accept stronger architectural assumptions to achieve the performance required for real-time monitoring.

\paragraph{Handling large System states.} An other challenge arises when the System possesses a large internal state (e.g., a massive database). Sharing the entire state in such settings is impractical. A potential solution is to augment our protocol with Private Information Retrieval (PIR). In this setup, the Monitor acts as a PIR client to retrieve specific shares of the state without revealing which indices it is accessing, while the System acts as the PIR server~\cite{CKGS98}. This extension would require additional cryptographic machinery beyond the primitives introduced in Section~\ref{sec:abstract-sharing} and is a promising direction for future research.

\paragraph{Integration with existing tools.} While our prototype demonstrates the feasibility of the core protocol, we have yet to integrate this setup with established runtime verification frameworks. Future work will focus on bridging our MP-SPDZ implementation with standard monitoring tools such as BeepBeep~\cite{MKH18}, DejaVu~\cite{HPU18}, or \textsc{MonPoly}~\cite{BKZ17}, allowing more familiar specification languages.

\paragraph{Robustness against active adversaries.} Although our primary evaluation focuses on the semi-honest setting, the protocol can be hardened against active (malicious) adversaries who may deviate from the specification. To prevent adaptive attacks where corrupt parties tailor their shares based on others' messages, we can enforce input commitments, requiring parties to cryptographically bind themselves to their shares before revealing them. To detect unauthorised modifications to stored values, the protocol can be lifted to authenticated secret sharing using information-theoretic Message Authentication Codes (MACs). These techniques however incur additional communication and computational overhead. 
Frameworks such as SPDZ~\cite{damgard2012multiparty} and
MASCOT~\cite{keller2016actively} provide efficient malicious-secure protocols
based on these principles.


\newpage
 \begin{credits}
 \subsubsection{\ackname}This work is a part of projects VAMOS that has received funding from the European Research Council (ERC), grant agreement No 101020093 and the Austrian Science Fund (FWF) SFB project SpyCoDe F8502.

 \end{credits}
%
%
%
\bibliographystyle{splncs04}
\bibliography{references}
\newpage
\appendix
\section{Instantiations of sharing systems}\label{app:SharingSystemsConcrete}
\subsection{Instantiations}
\label{sec:instantiations}

We now give several concrete instantiations of the abstract model.
These examples demonstrate that the abstraction captures a wide range of
practical secret-sharing schemes, each with different computational and
security trade-offs.

\paragraph{Additive sharing over a ring.}
Let $R$ be a finite ring (e.g., $\mathbb{Z}_{2^{64}}$ or $\mathbb{Z}_p$
for prime $p$).
Define $\mathcal{V} = \mathcal{D} = R$.
A value $v \in R$ is shared by choosing
$v^{(1)},\dots,v^{(k-1)} \xleftarrow{\$} R$
uniformly at random and setting
\[
  v^{(k)} = v - \sum_{i=1}^{k-1} v^{(i)} \pmod{R}.
\]
Reconstruction is summation modulo $R$:
\[
  \Reconstruct(v^{(1)},\dots,v^{(k)}) = \sum_{i=1}^k v^{(i)} \pmod{R}.
\]
This scheme supports addition and scalar multiplication locally
(without interaction), and multiplication via interactive protocols
such as Beaver multiplication~\cite{beaver1991efficient}.
The operation set $\mathcal{O}$ includes ring operations $\{+, -, \times\}$
and comparisons via conversion to Boolean form~\cite{catrina2010secure,damgard2006unconditionally}.

\paragraph{Polynomial (Shamir) sharing.}
Let $\mathbb{F}$ be a finite field and fix a threshold $t < k$.
Define $\mathcal{V} = \mathcal{D} = \mathbb{F}$.
A value $v \in \mathbb{F}$ is shared by sampling a random polynomial
$p(x) \in \mathbb{F}[x]$ of degree at most $t$ subject to $p(0)=v$,
and distributing shares $p(\alpha_1),\dots,p(\alpha_k)$,
where $\alpha_1,\dots,\alpha_k \in \mathbb{F}$ are distinct non-zero
evaluation points fixed in advance.
Reconstruction is Lagrange polynomial interpolation at $x=0$.
This scheme provides information-theoretic privacy
against any coalition of up to $t$ computers, and supports addition and
multiplication via Shamir's protocols~\cite{shamir1979share}.

\paragraph{Boolean sharing.}
Let $\mathcal{V} = \{0,1\}$ and $\mathcal{D} = \mathbb{Z}_2$.
A bit $b \in \{0,1\}$ is shared additively modulo $2$:
computers hold random bits $b^{(1)},\dots,b^{(k)}$ such that
$b = b^{(1)} \oplus \cdots \oplus b^{(k)}$.
This instantiation supports XOR locally and AND via garbled circuits~\cite{yao1986generate} or
other interactive protocols.
It is particularly suitable for Boolean circuits, branching, and predicate
evaluation.

\paragraph{Efficiency of share conversion.}
\begin{definition}[Efficient Share Conversion]
\label{def:efficient-conversion}
A share conversion
$\Convert_{\mathcal{S}_1 \to \mathcal{S}_2}$
is said to be \emph{efficient} if it can be realised by a protocol that:
\begin{enumerate}
  \item runs in a constant number of interaction rounds,
        independent of the value being converted, the number of computers,
        and the security parameter;
  \item has total communication complexity bounded by a polynomial in
        $\log|\mathcal{V}_1 \cap \mathcal{V}_2|$, the number of computers $k$,
        and the security parameter $\lambda$.
\end{enumerate}
\end{definition}

Intuitively, efficiency means that share conversion can be used as a
building block inside each protocol round without asymptotically dominating
the cost of the overall computation and without requiring
reconstruction of the underlying secret.

\paragraph{Share conversion examples.}
Typical efficient conversions include:
\begin{itemize}
  \item \emph{Arithmetic-to-Boolean conversion (bit decomposition):}
        decomposing a shared integer $\llbracket x \rrbracket$ modulo $2^n$ 
        into $n$ shared bits $\llbracket b_0 \rrbracket, \ldots, \llbracket b_{n-1} \rrbracket$
        such that $x = \sum_{i=0}^{n-1} b_i \cdot 2^i$.
        This enables bitwise operations and comparisons.
        Modern protocols~\cite{rathee2020cryptflow2,dalskov2020secure} achieve
        constant-round complexity using precomputed correlated randomness,
        with total communication $O(n)$ field elements per conversion.
        See \cref{sec:nonlinear} for protocol details.
        
  \item \emph{Boolean-to-arithmetic conversion:}
        embedding shared bits $\llbracket b_0 \rrbracket, \ldots, \llbracket b_{n-1} \rrbracket$
        into an arithmetic share $\llbracket x \rrbracket = \llbracket \sum_i b_i \cdot 2^i \rrbracket$.
        This can be done in constant rounds with $O(n \log n)$ communication
        using parallel prefix networks~\cite{demmler2015aby},
        or $O(n)$ communication with $O(\log n)$ rounds.
        This allows arithmetic operations on Boolean values.
        
  \item \emph{Cross-domain arithmetic conversion:}
        converting between different arithmetic domains,
        e.g., from $\mathbb{Z}_{2^k}$ to $\mathbb{Z}_p$ for prime $p$,
        or between additive and Shamir sharings.
        These conversions typically require $O(1)$ rounds
        and $O(\log |\mathcal{V}|)$ communication per value~\cite{cramer2018spdz2k}.
\end{itemize}
All of the above conversions admit efficient realisations in the sense of
Definition~\ref{def:efficient-conversion}, and satisfy the correctness and
privacy properties of Definition~\ref{def:share-conversion}.
See~\cite{demmler2015aby,mohassel2018aby3} for representative constructions.
Modern frameworks such as
SPDZ~\cite{damgard2012multiparty,damgard2013practical},
Overdrive~\cite{keller2018overdrive}, and
MP-SPDZ~\cite{keller2020mp}
provide efficient implementations of these primitives.

\subsection{Non-linear operations}
\label{sec:nonlinear}

While addition and scalar multiplication can be performed locally on additive shares,
non-linear operations such as comparison, equality testing, and multiplication
require interaction and often benefit from converting between sharing representations.

\paragraph{Comparison via bit decomposition.}
To evaluate $\llbracket x \rrbracket < \llbracket y \rrbracket$ where $x, y \in \mathbb{Z}_{2^n}$
are arithmetic-shared values, note that the relation $x < y$ does not preserve under addition
of random shares (i.e., $x^{(i)} < y^{(i)}$ does not imply $x < y$).
One approach~\cite{catrina2010secure,damgard2006unconditionally} converts to Boolean representation:

\begin{enumerate}
  \item \textbf{Bit decomposition:}
        Convert both arithmetic sharings to Boolean form.
        For a shared value $\llbracket x \rrbracket$, produce shared bits
        $\llbracket b_0 \rrbracket, \ldots, \llbracket b_{n-1} \rrbracket$
        (each a Boolean sharing in $\mathbb{Z}_2$) such that
        $x = \sum_{i=0}^{n-1} b_i \cdot 2^i$.
        
        Modern protocols achieve this in constant rounds with $O(n)$ communication
        per value~\cite{rathee2020cryptflow2,dalskov2020secure}.
        The key idea is to precompute correlated randomness offline:
        generate random Boolean sharings $\llbracket r_0 \rrbracket, \ldots, \llbracket r_{n-1} \rrbracket$
        along with their arithmetic sum $\llbracket r \rrbracket = \llbracket \sum_i r_i \cdot 2^i \rrbracket$.
        
        In the online phase, compute and open $c = x + r \bmod 2^n$.
        Since $r$ is uniformly random, $c$ reveals nothing about $x$.
        The computers then locally compute the public bits $c_0, \ldots, c_{n-1}$ of $c$,
        and use the precomputed randomness to compute
        $\llbracket b_i \rrbracket = c_i - \llbracket r_i \rrbracket - \llbracket \text{carry}_{i-1} \rrbracket$,
        where $\llbracket \text{carry}_{i-1} \rrbracket$ denotes the carry bit from position $i-1$
        in the binary addition $x + r$,
        via a constant-round protocol.
        
        Earlier protocols based on sequential carry propagation~\cite{catrina2010secure}
        required $O(n)$ rounds, but are now superseded by constant-round constructions.

  \item \textbf{Bitwise comparison:}
        With both $x$ and $y$ in Boolean form, evaluate a Boolean circuit
        that computes the comparison bit by bit, starting from the most significant bit.
        The circuit has depth $O(\log n)$ and size $O(n)$,
        yielding $O(\log n)$ rounds and $O(n)$ communication~\cite{damgard2006unconditionally}.
        The circuit outputs a single Boolean share $\llbracket x < y \rrbracket$.

  \item \textbf{Optional conversion back:}
        If the result is needed in arithmetic form (e.g., for conditional operations),
        convert the Boolean share to an arithmetic share in $O(1)$ rounds
        with $O(1)$ communication via Boolean-to-arithmetic conversion.
\end{enumerate}

\noindent
\textbf{Overall complexity:}
The entire comparison requires $O(\log n)$ rounds and $O(n)$ communication per comparison,
where $n$ is the bit length of the values being compared.

\paragraph{Equality testing.}
Equality testing $\llbracket x \rrbracket \stackrel{?}{=} \llbracket y \rrbracket$ can be performed more
efficiently than general comparison.
After bit decomposition (which dominates the cost at $O(1)$ rounds, $O(n)$ communication),
compute the bitwise XOR of corresponding bits:
$\llbracket d_i \rrbracket = \llbracket x_i \rrbracket \oplus \llbracket y_i \rrbracket$
for $i = 0, \ldots, n-1$.
This is a local operation (no communication).
Then $x = y$ if and only if all $d_i = 0$, which can be tested by computing
the OR-tree: $\llbracket d_0 \vee d_1 \vee \cdots \vee d_{n-1} \rrbracket$.
Using a binary tree of OR gates, this requires $O(\log n)$ rounds and $O(n)$ multiplications
in $\mathbb{Z}_2$ (AND gates can implement OR via De Morgan's laws).
The output is $0$ (true) if $x = y$, and $1$ (false) otherwise.

This approach has the same asymptotic complexity as comparison
($O(\log n)$ rounds, $O(n)$ communication),
but with smaller constants since it avoids the full comparison circuit.

\paragraph{Multiplication in arithmetic domains.}
Multiplication $\llbracket z \rrbracket = \llbracket x \rrbracket \times \llbracket y \rrbracket$
in additive or Shamir sharing requires interaction.
The standard approach uses \emph{Beaver triples}~\cite{beaver1991efficient}:
precomputed sharings $\llbracket a \rrbracket, \llbracket b \rrbracket, \llbracket c \rrbracket$
where $c = a \times b$ and $a, b$ are uniformly random.
Given such a triple, the computers open
$\llbracket e \rrbracket = \llbracket x - a \rrbracket$ and
$\llbracket f \rrbracket = \llbracket y - b \rrbracket$
in one round,
then locally compute
\[
  \llbracket z \rrbracket = \llbracket c \rrbracket + e \cdot \llbracket b \rrbracket 
  + f \cdot \llbracket a \rrbracket + e \cdot f.
\]
Since $e$ and $f$ are masked by random values, they reveal nothing about $x$ or $y$.
This requires $O(1)$ rounds and $O(1)$ field elements of communication per multiplication,
assuming Beaver triples are precomputed offline~\cite{damgard2012multiparty}.

\section{Operational semantics for our protocol}\label{app:Register}

\paragraph{Typed registers.}
The Monitor maintains a finite collection of abstract \emph{storage locations}
$R_0, R_1, \dots, R_{n-1}$, which we call \emph{registers} by analogy with
machine instruction sets.
Each register is a logical storage unit that holds a shared value under
the agreed sharing system $\mathcal{S}$.
Concretely, register $R_i$ represents a tuple of shares $(s_i^{(1)}, \dots, s_i^{(k)})$
distributed across the $k$ Monitor computers, where computer $j$ stores share $s_i^{(j)}$.
These are not physical CPU registers, but rather abstract memory locations in
the distributed computation.

Each register has a \emph{sharing type} indicating which representation is
currently in use:
\begin{itemize}
  \item \textbf{Arithmetic} ($\mathtt{Arith}$):
        the register holds an additive or polynomial sharing of a ring/field element.
        Supports operations $\{+, -, \times\}$ natively.
  \item \textbf{Boolean} ($\mathtt{Bool}$):
        the register holds a Boolean sharing (XOR-based).
        Supports operations $\{\land, \lor, \oplus, \neg\}$ natively.
\end{itemize}
We write $R_i : \mathtt{Arith}$ to denote that register $R_i$ currently holds
an arithmetic sharing, and similarly for $\mathtt{Bool}$.

At the beginning of round $t$, the specification state $\llbracket \mu_t \rrbracket$
is loaded into a designated subset of registers, and the freshly shared
observable output $\llbracket x_t \rrbracket$ is loaded into other registers.
All registers are initialised with a consistent sharing type determined by
the specification's requirements.

\paragraph{Instruction set.}
The Monitor evaluates the specification by executing a sequence of typed
instructions.
Each instruction operates on shared values without ever reconstructing them,
except for the final violation flag.

The instruction set includes:
\begin{itemize}
  \item \textbf{Arithmetic operations}
        (all operands and result must be $\mathtt{Arith}$):
        \begin{align*}
          &\mathtt{ADD}\; R_i, R_j, R_k
            &&\text{compute } \llbracket R_i \rrbracket \gets \llbracket R_j + R_k \rrbracket \\
          &\mathtt{SUB}\; R_i, R_j, R_k
            &&\text{compute } \llbracket R_i \rrbracket \gets \llbracket R_j - R_k \rrbracket \\
          &\mathtt{MUL}\; R_i, R_j, R_k
            &&\text{compute } \llbracket R_i \rrbracket \gets \llbracket R_j \times R_k \rrbracket \\
          &\mathtt{ADDC}\; R_i, R_j, c
            &&\text{compute } \llbracket R_i \rrbracket \gets \llbracket R_j + c \rrbracket
            \text{ for public } c
        \end{align*}

  \item \textbf{Boolean operations}
        (all operands and result must be $\mathtt{Bool}$):
        \begin{align*}
          &\mathtt{XOR}\; R_i, R_j, R_k
            &&\text{compute } \llbracket R_i \rrbracket \gets \llbracket R_j \oplus R_k \rrbracket \\
          &\mathtt{AND}\; R_i, R_j, R_k
            &&\text{compute } \llbracket R_i \rrbracket \gets \llbracket R_j \land R_k \rrbracket \\
          &\mathtt{NOT}\; R_i, R_j
            &&\text{compute } \llbracket R_i \rrbracket \gets \llbracket \neg R_j \rrbracket
        \end{align*}

  \item \textbf{Comparison operations}
        (operands must be $\mathtt{Arith}$, result is $\mathtt{Bool}$):
        \begin{align*}
          &\mathtt{LT}\; R_i, R_j, R_k
            &&\text{compute } \llbracket R_i \rrbracket_{\mathtt{Bool}} \gets \llbracket R_j < R_k \rrbracket \\
          &\mathtt{EQ}\; R_i, R_j, R_k
            &&\text{compute } \llbracket R_i \rrbracket_{\mathtt{Bool}} \gets \llbracket R_j = R_k \rrbracket
        \end{align*}
        \emph{This is where share conversion happens.}
        Internally, comparison operations perform share conversion to an
        intermediate representation suitable for comparison.
        The specific conversion depends on the underlying MPC protocol:
        it may involve converting arithmetic sharings to Boolean sharings
        (e.g., decomposing integers into shared bits for bitwise comparison),
        or converting between different arithmetic representations
        (e.g., from additive to polynomial sharings).
        The comparison protocol operates on these converted shares and produces
        a single Boolean-shared result bit $\llbracket R_i \rrbracket_{\mathtt{Bool}}$.
        
        The conversion protocol $\Convert_{\mathcal{S}_{\mathtt{Arith}} \to \mathcal{S}_{\mathtt{Bool}}}$
        (or between arithmetic representations)
        from Section~\ref{sec:share-conversion} is invoked internally by these
        comparison instructions.

  \item \textbf{Reconstruction}
        (only used for the final violation flag):
        \begin{align*}
          &\mathtt{REVEAL}\; b, R_i
            &&\text{reconstruct } b \gets \Reconstruct(\llbracket R_i \rrbracket)
        \end{align*}
        Here $b$ is a public output bit visible to all Monitor computers.
        This is the \emph{only} instruction that reveals a secret value.
\end{itemize}

\paragraph{Type discipline.}
Each instruction has a type signature that constrains which sharing types are
permitted for its operands and results.
A sequence of instructions is \emph{well-typed} if every instruction is applied
to registers whose current sharing types match the instruction's requirements.

The specification compiler (which translates the high-level specification
$(\phi, \delta)$ into instruction sequences) must ensure that the generated
instruction sequence is well-typed.
In this protocol, \emph{share conversion occurs exclusively within comparison
operations} ($\mathtt{LT}$, $\mathtt{EQ}$), which internally convert
arithmetic-shared operands to an intermediate form suitable for the comparison
protocol.

This makes the cost of share conversion \emph{explicit and measurable}:
by counting the number of comparison operations, we obtain a precise accounting
of the protocol's conversion overhead.

\paragraph{Example: evaluating $\phi(\mu_t, x_t) = (x_t > \mu_t)$.}
Suppose $\llbracket \mu_t \rrbracket$ is in $R_0 : \mathtt{Arith}$ and
$\llbracket x_t \rrbracket$ is in $R_1 : \mathtt{Arith}$.
The instruction sequence is:
\begin{align*}
  &\mathtt{LT}\; R_2, R_0, R_1
    &&\text{compute } \llbracket \mu_t < x_t \rrbracket \to R_2 : \mathtt{Bool}
    \\
  &\mathtt{REVEAL}\; \textit{flag}, R_2
    &&\text{output the violation flag}
\end{align*}
The $\mathtt{LT}$ instruction performs the necessary share conversions
internally: it converts both $\llbracket \mu_t \rrbracket_{\mathtt{Arith}}$
and $\llbracket x_t \rrbracket_{\mathtt{Arith}}$ to whatever intermediate
representation the comparison protocol requires,
executes the comparison, and outputs a single
Boolean-shared result bit in $R_2$.

\paragraph{Example: state update $\mu_{t+1} = \mu_t + x_t \bmod 256$.}
Suppose the specification maintains a running sum modulo $256$.
Starting with $\llbracket \mu_t \rrbracket$ in $R_0 : \mathtt{Arith}$ and
$\llbracket x_t \rrbracket$ in $R_1 : \mathtt{Arith}$:
\begin{align*}
  &\mathtt{ADD}\; R_3, R_0, R_1
    &&\llbracket \mu_{t+1} \rrbracket \gets \llbracket \mu_t + x_t \rrbracket
\end{align*}
No share conversion is needed since all values are arithmetic-shared.
Register $R_3$ holds $\llbracket \mu_{t+1} \rrbracket$ for the next round.

\paragraph{Example: complex predicate $(x_t > 100) \land (\mu_t = 5)$.}
This requires comparisons and Boolean logic:
\begin{align*}
  &\mathtt{LT}\; R_4, 100, R_1
    &&\llbracket 100 < x_t \rrbracket \to R_4 : \mathtt{Bool}
    \\
  &\mathtt{EQ}\; R_5, R_0, 5
    &&\llbracket \mu_t = 5 \rrbracket \to R_5 : \mathtt{Bool}
    \\
  &\mathtt{AND}\; R_6, R_4, R_5
    &&\llbracket R_4 \land R_5 \rrbracket \to R_6 : \mathtt{Bool}
    \\
  &\mathtt{REVEAL}\; \textit{flag}, R_6
\end{align*}
Each comparison ($\mathtt{LT}$ and $\mathtt{EQ}$) internally performs the
necessary share conversions on its arithmetic-shared operands.
The $\mathtt{AND}$ operation then combines the Boolean results without
any conversion, since both inputs are already Boolean-shared.

\paragraph{Remark on practical implementation.}
Modern MPC frameworks such as MP-SPDZ~\cite{keller2020mp} use a similar
instruction-based model with typed storage and share conversions embedded
within comparison operations.
Our formalization makes this computational substrate explicit within the
protocol description, enabling precise reasoning about the cost and correctness
of specification evaluation.
In particular, by identifying comparison as the sole source of share conversion
cost, we can optimize specifications to minimize the number of comparisons performed
in each round.

\section{Compositionality and security}
\compositionality*
\begin{proof}
By Property~\ref{prop:state-invariant}, the only information revealed in round $t$
is the single bit $\phi(\mu_t, x_t)$.
By Property~\ref{prop:privacy} (privacy of the sharing system),
the shares of $\llbracket \mu_{t+1} \rrbracket$ are independent of $\mu_{t+1}$
when viewed by any strict subset of Monitor computers.
Therefore, the protocol can continue indefinitely without accumulating additional
information leakage beyond the sequence of violation flags.
\end{proof}

\systemprivacy*
\begin{proof}
Fix an adversary structure $(I_M)$ with $|I_M| < k_M$.
At each round $t$, the System sends shares of $x_t$ to the Monitor
computers.
By Property~\ref{prop:privacy} (privacy of the sharing system),
the joint distribution of shares $\{\Share(x_t)_i\}_{i \in I_M}$
received by corrupted Monitor computers is independent of $x_t$.
Formally, for any two possible outputs $x_t, x_t' \in \mathcal{V}$,
the corrupted computers' shares have identical distributions,
and thus reveal no information about which value was actually shared.

The only value reconstructed at each round is the flag $\phi(\mu_t, x_t)$.
This reconstruction reveals one bit of information about $x_t$, namely whether
the pair $(\mu_t, x_t)$ satisfies the violation condition.
No additional information about $\sigma_t$ (which contains $x_t$ and potentially
other private components) is revealed.

Since the Monitor computers perform all intermediate computations
on shared values using secure protocols that preserve
Property~\ref{prop:privacy}, no intermediate results leak
information.
By induction over rounds, the total information revealed to corrupted
Monitor computers is exactly the sequence of flags and nothing more.
\end{proof}
\monitorprivacy*
\begin{proof}[Proof of \cref{thm:monitor-privacy}]
The System entity does not participate in the Monitor's internal computations.
The System's role is limited to (1) sending shares of $x_t$ to Monitor computers,
and (2) receiving the reconstructed flag value $\phi(\mu_t, x_t)$ at the end
of each round.

The shares sent by the System are generated independently by the System itself
and reveal nothing about how the Monitor will process them.
The Monitor never sends shares or intermediate values back to the System.
The only information flow from Monitor to System is the single bit
$\phi(\mu_t, x_t)$, which the System would learn anyway (by definition of the
monitoring task).

Since $\mu_t$ is maintained in shared form by the Monitor entity and never
reconstructed (except implicitly through the flag),
and since corrupted System computers are not part of the Monitor entity,
they receive no shares of $\mu_t$ and learn nothing about it.
Similarly, the functions $\delta$ and $\phi$ are evaluated internally by the
Monitor using secure multiparty computation, and no information about their
structure is revealed to the System beyond the output bit.
\end{proof}
\section{Reactive monitoring template}\label{app:monitor-template}
Our implementation follows a reactive monitoring template that
receives system inputs at each round, updates internal state,
checks for violations, and reports faults.
The following code demonstrates the core structure:

\pagebreak

\begin{lstlisting}[
language=Python,
caption={Template for Reactive Monitoring Systems},
label={lst:monitor-template},
basicstyle=\ttfamily\small,
breaklines=true,
mathescape=true
]
from Compiler.library import *

program.use_edabit(True)

n, T = int(program.args[1]), int(program.args[2])
PORT_NUM = 14000

def spec_function(monitor_state, system_inputs):
    next_state = $\delta$(monitor_state, system_inputs)
    fault = $\varphi$(next_state)
    final_state = fault * monitor_state + (1 - fault) * next_state
    return fault, final_state

monitor_state = ...

listen_for_clients(PORT_NUM)
CLIENT_ID = accept_client_connection(PORT_NUM)

for t in range(T):
    system_inputs = sint.receive_from_client(n, CLIENT_ID)
    fault, monitor_state = spec_function(monitor_state, system_inputs)
    sint.reveal_to_clients([CLIENT_ID], [fault])

closeclientconnection(CLIENT_ID)
\end{lstlisting}

\end{document}